%% file: Co_Ne_Sr_PSGD_TAC.tex
\documentclass[12pt,onecolumn,draftcls]{IEEEtran}
\usepackage{graphicx}
\usepackage{epsfig}
\usepackage{subfigure}
\usepackage{amssymb}
\usepackage{amsbsy}
\usepackage{amsmath}
\usepackage{cite}
\usepackage{url}

\usepackage{color}
\usepackage{float}
\usepackage{times,amsmath,epsfig}
\usepackage{xspace,latexsym,syntonly}
\usepackage{amssymb}
\usepackage{amsfonts}
\usepackage{textcomp}
\usepackage{subfigure}
\usepackage{amsbsy}
\usepackage{cite}
\input{macros}

\begin{document}
\title{On Projected Stochastic Gradient Descent Algorithm with Weighted Averaging for Least Squares Regression}
\author{Kobi Cohen, Angelia Nedi\'c and R. Srikant
\thanks{Kobi Cohen is with the Department of Electrical and Computer Engineering, Ben-Gurion University of the Negev, Beer-Sheva 84105, Israel, email: yakovsec@bgu.ac.il}
\thanks{Angelia Nedi\'c and R. Srikant are with the Coordinated Science Laboratory, University of Illinois at Urbana-Champaign, IL 61801.
Email:$\left\{\mbox{angelia, rsrikant}\right\}$@illinois.edu}
\thanks{Part of this work was presented at the The 41st International Conference on Acoustics, Speech, and Signal Processing (ICASSP), 2016.}
\thanks{This work was supported by AFOSR MURI FA 9550--10--1--0573, ONR Grant N00014--13--1--003, NSF Grant CNS--1161404.}}
\date{}
\maketitle

%-------------abstract----------------------------------
%-------------------------------------------------------
\begin{abstract}
\label{sec:abstract}
The problem of least squares regression of a $d$-dimensional unknown parameter is considered. A stochastic gradient descent based algorithm with weighted iterate-averaging that uses a single pass over the data is studied and its convergence rate is analyzed. We first consider a bounded constraint set of the unknown parameter. Under some standard regularity assumptions, we provide an explicit $O(1/k)$ upper bound on the convergence rate, depending on the variance (due to the additive noise in the measurements) and the size of the constraint set. We show that the variance term dominates the error and decreases with rate $1/k$, while the term which is related to the size of the constraint set decreases with rate $\log k/k^2$. We then compare the asymptotic ratio $\rho$ between the convergence rate of the proposed scheme and the empirical risk minimizer (ERM) as the number of iterations approaches infinity. We show that $\rho\leq 4$ under some mild conditions for all $d\geq 1$. We further improve the upper bound by showing that $\rho\leq 4/3$ for the case of $d=1$ and unbounded parameter set. Simulation results demonstrate strong performance of the algorithm as compared to existing methods, and coincide with $\rho\leq 4/3$ even for large $d$ in practice.
\end{abstract}
\IEEEpeerreviewmaketitle

% Keyword section
\def\keywords{\vspace{.5em}
{\bfseries\textit{Index Terms}---\,\relax%
}}
\def\endkeywords{\par}
\keywords
Convex optimization, projected stochastic gradient descent, weighted averaging, empirical risk minimizer.

\section{Introduction}
\label{sec:introduction}

For large-scale optimization problems, it is often desirable to minimize an unknown objective under computational constraints. Stochastic Gradient Descent (SGD) is a popular optimization method in a variety of machine learning tasks when dealing with very large data or with data streams. Specifically, instead of computing the true gradient (which is often computationally expensive) as in a standard gradient descent algorithm, in SGD-based methods the gradient is approximated by a single (or few) sample at each iteration. Using stochastic approximation analysis, it has been shown that SGD converges almost surely to a global minimum when the objective function is convex (otherwise it converges to a local minimum) under an appropriate learning rate and some regularity conditions \cite{bottou1998online}.

In this paper, we consider the problem of least mean squares regression, in which a $d$-dimensional unknown parameter is desired to be estimated from streaming noisy measurements. Specifically, let $\boldsymbol{x}$, $y$ be random variables with values in $\mathbb{R}^d$, and $\mathbb{R}$, respectively, and let $\Omega\subseteq\mathbb{R}^d$ be a compact convex constraint set for the unknown parameter. It is desired to minimize the expected least squares loss:
\beq
\bea{l}
\label{eq:opt_LS_intro}
\displaystyle\min_{\boldsymbol{\omega}}E\left[||\boldsymbol{x}^T\boldsymbol{\omega}-y||^2\right] \vspace{0.3cm}\\
\mbox{subject to}\;\; \boldsymbol{\omega}\in\Omega\subseteq\mathbb{R}^d
\ena
\eeq
from the samples stream $(\boldsymbol{x}_k, y_k)$ at times $k=1, 2, ...$ Motivated by recent studies on accelerated methods of SGD-based algorithms, we focus on a projected SGD method with weighted iterate-averaging to solve (\ref{eq:opt_LS_intro}).

\subsection{Main Results}
\label{ssec:main}

Solving (\ref{eq:opt_LS_intro}) directly is computationally inefficient since it requires high storage memory for the entire data and high computational complexity. Thus, our goal is to solve (\ref{eq:opt_LS_intro}) efficiently so that the running time and space usage are small. Motivated by recent studies showing that using averaging of the estimated parameter accelerates the convergence of SGD-based algorithms, we propose and analyze a Projected SGD with Weighted Averaging (PSGD-WA) algorithm for solving (\ref{eq:opt_LS_intro}). Specifically, a projected SGD iterates are computed at each time $k$, where averaged iterates are computed as byproducts of the algorithm (but not used in the construction of the PSGD iterates). The averaging weights are specified in terms of the step-sizes that the algorithm uses such that recent measurements are given higher weights (see Section \ref{ssec:constrained} for details). Our main results are as follows: i) We consider a bounded constraint set of the unknown parameter and propose a PSGD-WA algorithm that requires a single pass over the data. The proposed step size has a general form\footnote{It should be noted that previous studies on PSGD algorithms with weighted averaging (see \cite{lacoste2012simpler ,nedic2014stochastic}) considered only a fixed form of the step size without tuning parameters.} of $c\frac{\gamma}{k+\gamma}$, where $c>0$, $\gamma\geq 1$ are tunable parameters; ii) in contrast to previous studies on PSGD algorithms with weighted averaging showing a general order $O(1/k)$ of the error rate, we provide an explicit finite sample upper bound on the error obtained by the proposed PSGD-WA algorithm, depending on the variance (due to the additive noise in the measurements) and the size of the constraint set. We show that the variance term dominates the error and decreases with rate $1/k$, while the term which is related to the diameter of the constraint set decreases with rate $\log k/k^2$; iii) we compare the asymptotic ratio $\rho$ between the convergence rate of the proposed PSGD-WA and the empirical risk minimizer (ERM) (which is the minimizer in the absence of computational constraints) as the number of iterations approaches infinity. We show that $\rho\leq 4$ for all $d\geq 1$ when the random components of $\boldsymbol{x}$ are identically distributed and uncorrelated. We further improve the upper bound by showing that $\rho\leq 4/3$ for the case of $d=1$ and $x_k=x$ for all $k$. Simulation results demonstrate strong performance of the algorithm as compared to existing methods, and coincide with $\rho\leq 4/3$ even for large $d$ in practice.

\subsection{Related Work}
\label{ssec:related}

SGD is a computationally efficient method for solving large-scale optimization problems when dealing with very large data or with data streams. Accelerating SGD-based algorithms using averaging techniques has been studied in past and more recent years in \cite{nemirovskii1978cezare, nesterov1983method, cao1985convergence, polyak1992acceleration, polyak2001random, ang2001new, nesterov2004introductory, shin2004variable, juditsky2008learning, tseng2008accelerated, zhang2008new, hwang2009variable, beck2009fast, nesterov2009primal, nemirovski2009robust, xiao2009dual, rakhlin2011making, ghadimi2012optimal, lan2012optimal, lacoste2012simpler, schmidt2013minimizing, bonnabel2013stochastic, nedic2014stochastic, defossez2015averaged}. In \cite{tseng2008accelerated}, Tseng has developed an accelerated SGD-based algorithm with iterate-averaging that achieves convergence rate of $1/k^2$ for problems where the objective function has Lipschitz continuous gradients. This rate is known to be the best in the class of convex functions with Lipshitz gradients \cite{nesterov2004introductory}, for which the first fast algorithm was originally constructed by Nesterov \cite{nesterov1983method} for unconstrained problems, and was extended recently by Beck and Teboulle in \cite{beck2009fast} to a larger class of problems. Ghadimi and Lan used averaging in \cite{ghadimi2012optimal} to develop an algorithm that has the rate $1/k^2$ when the objective function has Lipschitz continuous gradients, and rate $1/k$ when the objective function is strongly convex. Juditsky et al. \cite{juditsky2008learning} considered a mirror-descent algorithm with averaging to construct aggregate estimators with the best achievable learning rate. Averaging techniques for the mirror-descent algorithm for stochastic problems involving the sum of a smooth objective and a nonsmooth objective function have been studied by Lan in \cite{lan2012optimal}. Other related works are concerned with iterate-averaging for best achievable rate of stochastic subgradients methods \cite{nemirovski2009robust, rakhlin2011making}, as well as gradient-averaging \cite{ang2001new, shin2004variable, blatt2007convergent, zhang2008new, hwang2009variable, nesterov2009primal, xiao2009dual, schmidt2013minimizing} and a sort of momentum \cite{tseng1998incremental, roux2012stochastic}, in which the algorithm uses a sort of weighting over previous gradients (instead of the iterate minimizer) in the construction of the algorithm.

The averaged iterates considered in this paper are not used in the construction of the PSGD iterates, but only computed as byproducts of them (see Section \ref{ssec:constrained} for details). Such methods have been studied by Nemirovski and Yudin \cite{nemirovskii1978cezare} for convex-concave saddle-point problems, by Polyak and Juditsky \cite{polyak1992acceleration} for stochastic gradient approximations and by Polyak \cite{polyak2001random} for convex feasibility problems. In \cite{polyak1992acceleration}, an asymptotically optimal performance has been achieved. However, a finite sample analysis remained open. More recently, Lacoste-Julien et al. \cite{lacoste2012simpler} used this averaging approach for a projected stochastic subgradient method to achieve $1/k$ convergence rate for strongly convex functions. Nedi\'c and Lee \cite{nedic2014stochastic} used a similar form of this scheme for a more general projected stochastic subgradient method using Bregman distances, which achieves $1/k$ convergence rate for strongly convex functions, and $1/\sqrt{k}$ convergence rate for general convex functions.

In this paper we focus on the testing error (i.e., the expected error on unseen data) of regression from noisy measurements, in which the convergence rate deteriorates (varies from $1/k$ to $1/\sqrt{k}$ per-iterate). While accelerating methods cannot be made faster, they have ability to produce estimates with low-variance, which attracted much interest in recent years \cite{Johnson2013accelerating, frostig2014competing, shalev2014accelerated, defossez2015averaged}. We focus on the strongly convex case, in which $O(1/k)$ is the best attainable convergence rate \cite{frostig2014competing}. However, this convergence rate is only optimal in the limit of large samples, and in practice other non-dominant terms may come into play in the finite sample regime. In \cite{frostig2014competing}, Frostig et al. have developed a Streaming Stochastic Variance Reduced Gradient (Streaming SVRG) algorithm using a constant step size, inspired by the SVRG algorithm developed by Johnson and Zhang \cite{Johnson2013accelerating}, and provided a finite sample analysis for a general strongly convex regression problems. They showed that the asymptotic ratio $\rho$ between the convergence rate of the Streaming SVRG and the ERM algorithm approaches $\rho=1$ as the number of iterations approaches infinity. However, achieving $\rho=1$ requires the sample batch size to grow geometrically occasionally for gradient-computing, as well as setting the constant step size close to zero (which deteriorates performance in the finite sample regime). In \cite{defossez2015averaged}, Defossez and Bach have developed a SGD algorithm using a constant step size with averaging for least mean squares regression, and provided a finite sample analysis. They showed that $\rho=1$ as the constant step size is set close to zero, which deteriorates performance in the finite sample regime. In this paper, however, the proposed PSGD-WA algorithm uses decreasing step-sizes which can be large in the beginning of the algorithm and decrease as the number of iterations increases. The proposed PSGD-WA uses a weighted averaging of the estimates, by letting higher weights to recent measurements. We provide a finite sample analysis as well as an asymptotic upper bound $\rho\leq 4$ when $d\geq 1$ and $\rho\leq 4/3$ when $d=1$. Note that our results does not require the sample batch size to grow geometrically occasionally as in \cite{frostig2014competing} or setting small step-sizes in the beginning of the algorithm as in \cite{frostig2014competing, defossez2015averaged}. Thus, the proposed PSGD-WA algorithm is expected to perform well in the non-asymptotic case in addition to the nice asymptotic property as illustrated by simulation results provided in Section \ref{sec:simulation}.

\subsection{Notations}
\label{sec:notations}

Throughout the paper, small letters denote scalars, boldface small letters denote column vectors, and boldface capital letters denote matrices. All vectors are column vectors. The term $\mathbf{z}^T$ denotes the conjugate transpose of the vector $\mathbf{z}$, and $||\cdot||$ denotes the Euclidean norm. The subscript $k$ associated with a r.v. denotes a realization at time $k$.

\section{Problem Statement}
\label{sec:constrained}

Let $\boldsymbol{x}$, $y$ be random variables with values in $\mathbb{R}^d$, and $\mathbb{R}$, respectively. At each time $k$, we observe i.i.d. samples across time $(\boldsymbol{x}_k, y_k)$. We assume that $E\left[\boldsymbol{x}^T\boldsymbol{x}\right]$ is finite and we denote by $R_x=E\left[\boldsymbol{x}\boldsymbol{x}^T\right]$ the correlation matrix of $\boldsymbol{x}$.

It is desired to minimize the expected least squares loss:
\beq
\bea{l}
\label{eq:opt_LS}
\displaystyle\min_{\boldsymbol{\omega}}E\left[||\boldsymbol{x}^T\boldsymbol{\omega}-y||^2\right] \vspace{0.3cm}\\
\mbox{subject to}\;\; \boldsymbol{\omega}\in\Omega\subseteq\mathbb{R}^d
\ena
\eeq
from the samples stream $(\boldsymbol{x}_k, y_k)$ at times $k=1, 2, ...$. It is assumed that $R_x$ is invertible (i.e., strongly convex case).
%(see, e.g., Nesterov, 2004)).
We denote by $\mu$ the smallest eigenvalue of $R_x$, so that $\mu>0$.

We denote the optimal solution of (\ref{eq:opt_LS}) by $\boldsymbol{\omega}^*\in\mathbb{R}^d$, and it is assumed that a decision maker knows that $\boldsymbol{\omega}^*$ lies in the interior of a convex constraint set $\Omega\subseteq\mathbb{R}^d$.

Let
\begin{center}
$f(\boldsymbol{\omega})\triangleq E\left[||\boldsymbol{x}^T\boldsymbol{\omega}-y||^2\right]$
\end{center}
be the mean squares loss as a function of $\boldsymbol{\omega}$, and $f^*=f(\boldsymbol{\omega}^*)\in\mathbb{R}$ be the value at the minimum. The term $v_k=\boldsymbol{x}_k^T\boldsymbol{\boldsymbol{\omega}}^*-y_k$ denotes the zero-mean additive noise with variance $\sigma^2$.
The gradient of $f$ at $\boldsymbol{\omega}$ is defined by $\nabla\boldsymbol{f}\left(\boldsymbol{\omega}\right)
=E\left[2\boldsymbol{x}\left(\boldsymbol{x}^T\boldsymbol{\omega}-y\right)\right]
=E\left[\boldsymbol{g}_k(\boldsymbol{\omega})\right]$, where $\boldsymbol{g}_k(\boldsymbol{\omega})\triangleq 2\boldsymbol{x}_{k}\left(\boldsymbol{x}_{k}^T\boldsymbol{\omega}-y_{k}\right)$
is the estimate of the gradient at $\boldsymbol{\omega}$ based on a single sample at iteration $k$\footnote{When a few samples are available per iteration we estimate the gradient by averaging.}.
For convenience, we write $\nabla \boldsymbol{f}_k\triangleq\nabla \boldsymbol{f}\left(\boldsymbol{\omega}_k\right)$ and $\boldsymbol{g}_k\triangleq\boldsymbol{g}_k(\boldsymbol{\omega}_k)$ when referring to the gradients at $\boldsymbol{\omega}_k$, where $\boldsymbol{\omega}_k$ is the estimate of $\boldsymbol{\omega}$ at iteration $k$ obtained by an iterative algorithm (see the next section for details). The error at the $k^{th}$ iteration is defined by $\boldsymbol{e}_k\triangleq\boldsymbol{\omega}_k-\boldsymbol{\omega}^*$. Note that
\beq
\bea{l}
\label{eq:g_k_develop1}
\displaystyle
\boldsymbol{g}_k=2\boldsymbol{x}_{k}\left(\boldsymbol{x}_{k}^T\boldsymbol{\omega}_{k}-y_{k}\right)\vspace{0.3cm}\\
=2\boldsymbol{x}_{k}\left(\boldsymbol{x}_{k}^T\boldsymbol{\omega}_{k}-\boldsymbol{x}_{k}^T\boldsymbol{\omega}^*+v_{k}\right)\vspace{0.3cm}\\
=2\boldsymbol{x}_{k}\left(\boldsymbol{x}_{k}^T \boldsymbol{e}_k+v_{k}\right).
\ena
\eeq

\section{Projected Stochastic Gradient descent algorithm with Weighted Averaging}
\label{ssec:constrained}

We investigate a Projected Stochastic Gradient descent algorithm with Weighted Averaging (PSGD-WA). According to PSGD-WA, we hold two estimates of $\boldsymbol{\omega}^*$ at each iteration, denoted by $\boldsymbol{\omega}_k, \bar{\boldsymbol{\omega}}_k$. The estimate $\boldsymbol{\omega}_k$ is computed at each iteration (say $k$), and $\bar{\boldsymbol{\omega}}_k$ is the weighted average estimate based on all estimates up to time $k$. Let $\lambda_k$ be the step-size at time $k$, and assume that it diminishes with $k$.

Let $\boldsymbol{\omega}_0\in\Omega$ be an initial estimate of $\boldsymbol{\omega}$ (possibly random). At iteration $k=1$ we compute the projected estimate of $\boldsymbol{\omega}^*$ based on the random measurements $(\boldsymbol{x}_{0}, y_{0})$ and the initial estimate $\boldsymbol{\omega}_{0}$:
\beq
\label{eq:PSGD_WA1_init}
\bea{l}
\displaystyle\boldsymbol{\omega}_{1}=\arg\;\min_{\boldsymbol{\omega}\in\Omega}
\left\{\lambda_0 \boldsymbol{g}_{0}^T\cdot\left(\boldsymbol{\omega}-\boldsymbol{\omega}_0\right)+
\frac{1}{2}||\boldsymbol{\omega}-\boldsymbol{\omega}_0||^2\right\},
\ena
\eeq
and we compute this estimate iteratively. In general, at iteration $k+1$ we compute the projected estimate of $\boldsymbol{\omega}^*$ based on the random measurements $(\boldsymbol{x}_{k}, y_{k})$ and the last estimate $\boldsymbol{\omega}_{k}$:
\beq
\label{eq:PSGD_WA1}
\bea{l}
\displaystyle\boldsymbol{\omega}_{k+1}=\arg\;\min_{\boldsymbol{\omega}\in\Omega}
\left\{\lambda_k \boldsymbol{g}_{k}^T\cdot\left(\boldsymbol{\omega}-\boldsymbol{\omega}_k\right)+
\frac{1}{2}||\boldsymbol{\omega}-\boldsymbol{\omega}_k||^2\right\} \vspace{0.3cm}\\\hspace{6cm}\forall k\geq 0.
\ena
\eeq
It can be verified that $\boldsymbol{\omega}_{k+1}$ projects the unconstrained gradient descent iterate $\boldsymbol{\omega}_k-\lambda_k\boldsymbol{g}_{k}$ into $\Omega$. Motivated by previous studies on SGD with iterate-averaging (e.g., \cite{lacoste2012simpler, nedic2014stochastic}), in addition to the estimate $\boldsymbol{\omega}_{k+1}$, we propose to compute the weighted average estimate:
\beq
\label{eq:PSGD_WA2}
\bea{l}
\displaystyle\bar{\boldsymbol{\omega}}_{k+1}=\sum_{i=0}^{k+1}\beta_{k+1,i} \boldsymbol{\omega}_{i},
\ena
\eeq
where $\beta_{k,0}, \beta_{k,1}, ..., \beta_{k,k}$ are nonnegative scalars with the sum equals $1$, where the weighted average estimate $\displaystyle\bar{\boldsymbol{\omega}}_{k}$ is computed based on the first $k$ iterations. These convex weights will be defined in terms of the step size values $\lambda_0, \lambda_1, ..., \lambda_k$, and $\displaystyle\bar{\boldsymbol{\omega}}_{k}$ will be computed recursively (see (\ref{eq:recursive_constrained}) in Section \ref{ssec:implementation}). In Section \ref{ssec:performance_bounded} we will analyze the convergence rate of $\bar{\boldsymbol{\omega}}_{k}$ to the solution of (\ref{eq:opt_LS}).

\subsection{Implementation and Complexity Discussion}
\label{ssec:implementation}

The PSGD-WA algorithm is simple for implementation as compared to existing methods. At iteration $k$, the algorithm requires to store $\boldsymbol{\omega}_{k}$, the weighted average $\bar{\boldsymbol{\omega}}_{k-1}$ and the normalization term $S_{k-1}=\sum_{r=0}^{k-1}1/\alpha_r$. The weighted average $\bar{\boldsymbol{\omega}}_{k}$ can be updated recursively by computing
\begin{center}
$S_k=S_{k-1}+1/\alpha_k$
\end{center}
and then by setting:
\beq
\label{eq:recursive_constrained}
\displaystyle\bar{\boldsymbol{\omega}}_{k}=\frac{S_{k-1}}{S_k}\bar{\boldsymbol{\omega}}_{k-1}
+\left(1-\frac{S_{k-1}}{S_k}\right)\boldsymbol{\omega}_k.
\eeq
As a result, only $O(1)$ computations are required per iteration as needed by the classic SGD algorithm. Note that PSGD-WA does not require the sample batch size to grow as in \cite{frostig2014competing}. The storage memory required by PSGD-WA is similar to that required by the average SGD with constant step size algorithm proposed in \cite{defossez2015averaged}.

\section{Performance Analysis}
\label{ssec:performance_bounded}

In this section we analyze the algorithm's performance when the constraint set $\Omega$ is bounded. Let $e_{max}=\sup_{\boldsymbol{\omega}\in\Omega}\left\{||\boldsymbol{\omega}-\boldsymbol{\omega}^*||^2\right\}$ be the maximal square error of any projected estimate of $\boldsymbol{\omega}^*$. Let $\mathcal{F}_{k-1}=
\sigma\left\{\boldsymbol{\omega}_0, \boldsymbol{x}_0, y_0, \boldsymbol{x}_1, y_1, ..., \boldsymbol{x}_{k-1}, y_{k-1}\right\}$ be the filtration generated by the history of the algorithm starting at time $0$ up to time $k-1$. Note that $\boldsymbol{\omega}_0, \boldsymbol{\omega}_1, ..., \boldsymbol{\omega}_k$ are known once $\mathcal{F}_{k-1}$ is given\vspace{0.3cm}.

%---------------------lemma-----------------------
\begin{lemma}
\label{lemma:upper_bound_F_k_1}
Assume that (\ref{eq:PSGD_WA1}) is implemented. Then, for all $\boldsymbol{\omega}\in\Omega$ and $k\geq 0$, we have:
\beq
\bea{l}
\label{eq:upper_bound_F_k_1}
\displaystyle\frac{1}{2}E\left[||\boldsymbol{\omega}_{k+1}-\boldsymbol{\omega}||^2|\mathcal{F}_{k-1}\right]+\lambda_k \nabla\boldsymbol{f}_k^T\cdot\left(\boldsymbol{\omega}_k-\boldsymbol{\omega}\right) \vspace{0.3cm}\\\hspace{0.5cm}\displaystyle\leq
\frac{1}{2}||\boldsymbol{\omega}_k-\boldsymbol{\omega}||^2+
2\lambda_k^2E\left[||\boldsymbol{x}_{k}\boldsymbol{x}_{k}^T \boldsymbol{e}_k||^2|\mathcal{F}_{k-1}\right] \vspace{0.3cm}\\\hspace{3cm}\displaystyle
+2\lambda_k^2\sigma^2 E\left[||\boldsymbol{x}_{k}||^2\right]. \vspace{0.3cm}
\ena
\eeq
\end{lemma}
\begin{proof}
We first upper bound the term $\lambda_k \boldsymbol{g}_k^T\cdot\left(\boldsymbol{\omega}_{k+1}-\boldsymbol{\omega}\right)$. Since $\boldsymbol{\omega}_{k+1}$ solves (\ref{eq:PSGD_WA1}), we have:
\beq
\bea{l}
\displaystyle
\nabla_{\boldsymbol{\omega}}q(\boldsymbol{\omega}_{k+1})^T\left(\boldsymbol{\omega}-\boldsymbol{\omega}_{k+1}\right)
\vspace{0.3cm}\\\displaystyle
=\left(\lambda_k\boldsymbol{g}_k+\boldsymbol{\omega}_{k+1}-\boldsymbol{\omega}_k\right)^T
    \left(\boldsymbol{\omega}-\boldsymbol{\omega}_{k+1}\right)\geq 0  \;\;\forall \boldsymbol{\omega}\in\Omega,
\ena
\eeq
where $q(\boldsymbol{\omega})=\lambda_k \boldsymbol{g}_{k}^T\cdot\left(\boldsymbol{\omega}-\boldsymbol{\omega}_k\right)+
\frac{1}{2}||\boldsymbol{\omega}-\boldsymbol{\omega}_k||^2$ is the objective function in (\ref{eq:PSGD_WA1}).
Arranging terms yields:
\beq
\bea{l}
\displaystyle
\lambda_k\boldsymbol{g}_k^T\left(\boldsymbol{\omega}_{k+1}-\boldsymbol{\omega}\right)\vspace{0.3cm}\\\displaystyle
\leq\left(\boldsymbol{\omega}_{k+1}-\boldsymbol{\omega}_k\right)^T\left(\boldsymbol{\omega}-\boldsymbol{\omega}_{k+1}\right)
\vspace{0.3cm}\\\displaystyle
=\frac{1}{2}||\boldsymbol{\omega}_k-\boldsymbol{\omega}||^2
 -\frac{1}{2}||\boldsymbol{\omega}_{k+1}-\boldsymbol{\omega}||^2
 -\frac{1}{2}||\boldsymbol{\omega}_k-\boldsymbol{\omega}_{k+1}||^2.
\ena
\eeq
Next, we lower bound the term $\lambda_k \boldsymbol{g}_k^T\cdot\left(\boldsymbol{\omega}_{k+1}-\boldsymbol{\omega}\right)$.
\beq
\bea{l}
\displaystyle
\lambda_k\boldsymbol{g}_k^T\left(\boldsymbol{\omega}_{k+1}-\boldsymbol{\omega}\right)\vspace{0.3cm}\\\displaystyle
=\lambda_k\boldsymbol{g}_k^T\left(\boldsymbol{\omega}_{k+1}-\boldsymbol{\omega}_k\right)
+\lambda_k\boldsymbol{g}_k^T\left(\boldsymbol{\omega}_k-\boldsymbol{\omega}\right)
\vspace{0.3cm}\\\displaystyle
\geq -\frac{\lambda_k^2}{2}||\boldsymbol{g}_k||^2-\frac{1}{2}||\boldsymbol{\omega}_{k+1}-\boldsymbol{\omega}_k||^2
        +\lambda_k\boldsymbol{g}_k^T\left(\boldsymbol{\omega}_k-\boldsymbol{\omega}\right).
\ena
\eeq
Finally, combining the lower and upper bounds on $\lambda_k \boldsymbol{g}_k^T\cdot\left(\boldsymbol{\omega}_{k+1}-\boldsymbol{\omega}\right)$ yields:
\beq
\bea{l}
\displaystyle
\lambda_k \boldsymbol{g}_k^T\cdot\left(\boldsymbol{\omega}_k-\boldsymbol{\omega}\right)
\vspace{0.3cm}\\\displaystyle
\leq\frac{1}{2}||\boldsymbol{\omega}_k-\boldsymbol{\omega}||^2
-\frac{1}{2}||\boldsymbol{\omega}_{k+1}-\boldsymbol{\omega}||^2
+\frac{\lambda_k^2}{2}||\boldsymbol{g}_k||^2.
\ena
\eeq
Taking expectation conditioned on $\mathcal{F}_{k-1}$ yields:
\beq
\bea{l}
\displaystyle
\lambda_k \nabla\boldsymbol{f}_k^T\cdot\left(\boldsymbol{\omega}_k-\boldsymbol{\omega}\right)
\vspace{0.3cm}\\\displaystyle
\leq\frac{1}{2}||\boldsymbol{\omega}_k-\boldsymbol{\omega}||^2
-\frac{1}{2}E\left[||\boldsymbol{\omega}_{k+1}-\boldsymbol{\omega}||^2|\mathcal{F}_{k-1}\right] \vspace{0.3cm}\\\displaystyle\hspace{2cm}
+\frac{\lambda_k^2}{2}E\left[||\boldsymbol{g}_k||^2|\mathcal{F}_{k-1}\right].
\ena
\eeq
where we used the fact that $E\left[\boldsymbol{g}_k|\mathcal{F}_{k-1}\right]=\nabla\boldsymbol{f}_k$, and $\boldsymbol{\omega}_k$ is deterministic conditioned on $\mathcal{F}_{k-1}$. Finally, using (\ref{eq:g_k_develop1}) we have $E\left[||\boldsymbol{g}_k||^2|\mathcal{F}_{k-1}\right]\leq
4 E\left[||\boldsymbol{x}_{k}\boldsymbol{x}_{k}^T \boldsymbol{e}_k||^2|\mathcal{F}_{k-1}\right]
+4 \sigma^2 E\left[||\boldsymbol{x}_{k}||^2\right]$. Thus, (\ref{eq:upper_bound_F_k_1}) follows. \vspace{0.3cm}
\end{proof}

Next, Consider a sequence
\beq
\alpha_k=\frac{\gamma}{\gamma+k} , \;\;\; k=0, 1, ... \vspace{0.3cm}
\eeq
%---------------------lemma step size-----------------------
\begin{lemma}
\label{le_step_size}
The sequence $\alpha_k$, with $\gamma\geq 2$ satisfies:
\beq
\bea{l}
\label{eq:le_step_size}
\displaystyle
\alpha_k^2\geq\frac{1}{\sum_{r=0}^{k}1/\alpha_r} \;\;\; \forall k\geq 0. \vspace{0.3cm}
\ena
\eeq
\end{lemma}
\begin{proof}
Note that it suffices to show that the step size satisfies:
\beq
\label{eq:pr_le_step_size1}
\frac{1}{\alpha_{r+1}^2}-\frac{1}{\alpha_r^2}\leq\frac{1}{\alpha_{r+1}}
\eeq
for $r=0, 1, ...$, since summing (\ref{eq:pr_le_step_size1}) over $r=0, 1, k-1$ yields: $\frac{1}{\alpha_k^2}-\frac{1}{\alpha_0^2}\leq\sum_{r=1}^{k}\frac{1}{\alpha_r}$, which yields (\ref{eq:le_step_size}).
%Variations of condition (\ref{eq:pr_le_step_size1}) have been used in \cite{Nesterov, Tseng, Nedich}.

Next, we show that the step size with $\gamma\geq 2$ satisfies (\ref{eq:pr_le_step_size1}) for $r\geq 0$. Note that  (\ref{eq:pr_le_step_size1}) can be written as $\frac{1-\alpha_{r+1}}{\alpha_{r+1}^2}\leq\frac{1}{\alpha_r^2}$, where substituting $\alpha_r=\frac{\gamma}{\gamma+r}$ in the last inequality yields:
\beq
\displaystyle
\frac{1-\frac{\gamma}{\gamma+r+1}}{\left[\frac{\gamma}{\gamma+r+1}\right]^2}
\leq\frac{1}{\left[\frac{\gamma}{\gamma+r}\right]^2}.
\eeq
After some algebraic manipulations we obtain the following quadratic inequality:
\beq
\label{eq:quadratic}
\displaystyle
\gamma^2+(r-1)\gamma-2r-1\geq 0.
\eeq
The solution for (\ref{eq:quadratic}) yields:
\beq
\displaystyle
\gamma\geq \gamma(r)\triangleq\frac{-r+1+\sqrt{r^2+6r+5}}{2} \;\;\;\forall r\geq 0.
\eeq
Thus, setting $\alpha_r=\frac{\tilde{\gamma}(r)}{\tilde{\gamma}(r)+r}$ with $\tilde{\gamma}(r)\geq\gamma(r)$ satisfies (\ref{eq:le_step_size}) for all $r\geq 0$. Next, it can be verified that $\gamma(r)$ is monotonically increasing for all $r\geq 0$ and has limit $\lim_{r\rightarrow\infty}\gamma(r)=2$. Thus, $\gamma(r)\leq 2$ for all $r\geq 0$. Hence, setting $\gamma\geq 2$ is sufficient to satisfy (\ref{eq:le_step_size}) for all $r\geq 0$. \vspace{0.3cm}
\end{proof}

%---------------------theorem-----------------------
\begin{theorem}
\label{th:main}
Assume that PSGD-WA is implemented, with
\beq
\bea{l}
\label{eq:th_lambda_beta}
\displaystyle\lambda_k=\frac{1}{2\mu}\alpha_k=\frac{1}{2\mu}\frac{\gamma}{\gamma+k}\vspace{0.3cm}\\
\displaystyle\beta_{k,i}=\frac{1/\alpha_i}{\sum_{r=0}^k 1/\alpha_r}, \vspace{0.3cm}
\ena
\eeq
where $\gamma\geq 2$. Then, for all $k\geq 0$ we have:
\beq
\bea{l}
\label{eq:th_mean_squares1}
\displaystyle
E\left[f\left(\bar{\boldsymbol{\omega}}_k\right)\right]-f\left(\boldsymbol{\omega}^*\right)
\vspace{0.3cm}\\\hspace{0.0cm}\displaystyle
\leq
\frac{\left(\log(k+1)+1\right)\gamma^2 E\left[||\boldsymbol{x}_{k}\boldsymbol{x}_{k}^T||^2\right]C^2}{\mu^2(\gamma+k)^2} \vspace{0.3cm}\\\hspace{2cm}\displaystyle
+\frac{(k+1)\gamma^2E\left[||\boldsymbol{x}_{k}||^2\right]\sigma^2}{\mu(\gamma+k)^2}, \;\;\forall \gamma\geq2 \;\forall k\geq 0,
\ena
\eeq
where
\beq
\label{eq:C_def}
C^2\triangleq 4e_{max}d E\left[||\boldsymbol{x}_{k}\boldsymbol{x}_{k}^T||^2\right]+4\sigma^2E\left[||\boldsymbol{x}_{k}||^2\right]. \vspace{0.3cm}\\
 \eeq
\end{theorem}
%--------------proof-------------------
\begin{proof}
By Lemma \ref{lemma:upper_bound_F_k_1}, setting $\boldsymbol{\omega}=\boldsymbol{\omega}^*$ in (\ref{eq:upper_bound_F_k_1}) yields:
\beq
\bea{l}
\label{eq:applying_previous_lemma}
\displaystyle\frac{1}{2}E\left[||\boldsymbol{\omega}_{k+1}-\boldsymbol{\omega}^*||^2|\mathcal{F}_{k-1}\right]+\frac{\alpha_k}{2\mu} \nabla\boldsymbol{f}_k^T\cdot\left(\boldsymbol{\omega}_k-\boldsymbol{\omega}^*\right) \vspace{0.3cm}\\\hspace{0.5cm}\displaystyle\leq
\frac{1}{2}||\boldsymbol{\omega}_k-\boldsymbol{\omega}^*||^2+
\frac{\alpha_k^2}{2\mu^2}E\left[||\boldsymbol{x}_{k}\boldsymbol{x}_{k}^T \boldsymbol{e}_k||^2|\mathcal{F}_{k-1}\right] \vspace{0.3cm}\\\hspace{3cm}\displaystyle
+\frac{\alpha_k^2}{2\mu^2}\sigma^2 E\left[||\boldsymbol{x}_{k}||^2\right].
\ena
\eeq
Note that $2\mu$-strong convexity of $f$ implies:
\beq
\label{eq:pr_strong_convexity}
\displaystyle
\nabla\boldsymbol{f}_k^T\cdot\left(\boldsymbol{\omega}_k-\boldsymbol{\omega}^*\right)
\geq f\left(\boldsymbol{\omega}_k\right)-f\left(\boldsymbol{\omega}^*\right)+\mu||\boldsymbol{\omega}_k-\boldsymbol{\omega}^*||^2.
\vspace{0.3cm}
\eeq
Substituting (\ref{eq:pr_strong_convexity}) in (\ref{eq:applying_previous_lemma}) and smoothing yields:
\beq
\bea{l}
\label{eq:upper_bound_smoothing}
\displaystyle\frac{1}{2}E\left[||\boldsymbol{\omega}_{k+1}-\boldsymbol{\omega}^*||^2\right]+\frac{\alpha_k}{2\mu}
\left(E\left[f\left(\boldsymbol{\omega}_k\right)\right]-f\left(\boldsymbol{\omega}^*\right)\right)
\vspace{0.3cm}\\\hspace{0.0cm}\displaystyle
\overset{(a)}\leq
\frac{1}{2}E\left[||\boldsymbol{\omega}_{k+1}-\boldsymbol{\omega}^*||^2\right]
+\frac{\alpha_k}{2\mu}\nabla\boldsymbol{f}_k^T\cdot\left(\boldsymbol{\omega}_k-\boldsymbol{\omega}^*\right)
\vspace{0.3cm}\\\hspace{4cm}\displaystyle
-\frac{\alpha_k}{2}E\left[||\boldsymbol{\omega}_{k}-\boldsymbol{\omega}^*||^2\right]
\vspace{0.3cm}\\\hspace{0.0cm}\displaystyle
\overset{(b)}\leq
\frac{1}{2}E\left[||\boldsymbol{\omega}_k-\boldsymbol{\omega}^*||^2\right]+
\frac{\alpha_k^2}{2\mu^2}E\left[||\boldsymbol{x}_{k}\boldsymbol{x}_{k}^T \boldsymbol{e}_k||^2\right] \vspace{0.3cm}\\\hspace{0.5cm}\displaystyle
+\frac{\alpha_k^2}{2\mu^2}\sigma^2 E\left[||\boldsymbol{x}_{k}||^2\right]
-\frac{\alpha_k}{2}E\left[||\boldsymbol{\omega}_{k}-\boldsymbol{\omega}^*||^2\right]
\vspace{0.3cm}\\\hspace{0.0cm}\displaystyle
=\frac{1-\alpha_k}{2}E\left[||\boldsymbol{\omega}_k-\boldsymbol{\omega}^*||^2\right]+
\frac{\alpha_k^2}{2\mu^2}E\left[||\boldsymbol{x}_{k}\boldsymbol{x}_{k}^T \boldsymbol{e}_k||^2\right] \vspace{0.3cm}\\\hspace{3cm}\displaystyle
+\frac{\alpha_k^2}{2\mu^2}\sigma^2 E\left[||\boldsymbol{x}_{k}||^2\right]
\vspace{0.3cm}\\\hspace{0cm}\displaystyle
\leq\frac{1-\alpha_k}{2}E\left[||\boldsymbol{\omega}_k-\boldsymbol{\omega}^*||^2\right]+
\frac{\alpha_k^2}{2\mu^2}E\left[||\boldsymbol{x}_{k}\boldsymbol{x}_{k}^T||^2\right]E\left[||\boldsymbol{e}_k||^2\right] \vspace{0.3cm}\\\hspace{3cm}\displaystyle
+\frac{\alpha_k^2}{2\mu^2}\sigma^2 E\left[||\boldsymbol{x}_{k}||^2\right].
\ena
\eeq
Inequality $(a)$ follows by (\ref{eq:pr_strong_convexity}), and inequality $(b)$ follows by (\ref{eq:applying_previous_lemma}).
Next, we upper bound $E\left[||\boldsymbol{e}_k||^2\right]$. Note that $E\left[||\boldsymbol{g}_k||^2\right]$ is bounded by $E\left[||\boldsymbol{g}_k||^2\right]=E\left[||2\boldsymbol{x}_{k}\left(\boldsymbol{x}_{k}^T \boldsymbol{e}_k+v_{k}\right)||^2\right]\leq C^2$. Thus, using a similar argument as in \cite[Theorem 1]{nedic2014stochastic} yields:
\beq
\label{eq:upper_bound_err_nedich}
E\left[||\boldsymbol{e}_k||^2\right]\leq\frac{C^2}{(k+1)\mu}.
\eeq
As a result, substituting (\ref{eq:upper_bound_err_nedich}) in (\ref{eq:upper_bound_smoothing}) yields:
\beq
\bea{l}
\label{eq:upper_bound_smoothing2}
\displaystyle\frac{1}{2}E\left[||\boldsymbol{\omega}_{k+1}-\boldsymbol{\omega}^*||^2\right]
+\frac{\alpha_k}{2\mu}\left(E\left[f\left(\boldsymbol{\omega}_k\right)\right]-f\left(\boldsymbol{\omega}^*\right)\right)
\vspace{0.3cm}\\\hspace{0.0cm}\displaystyle
\leq\frac{1-\alpha_k}{2}E\left[||\boldsymbol{\omega}_k-\boldsymbol{\omega}^*||^2\right]+
\frac{\alpha_k^2C^2}{2\mu^3(k+1)}E\left[||\boldsymbol{x}_{k}\boldsymbol{x}_{k}^T||^2\right] \vspace{0.3cm}\\\hspace{3cm}\displaystyle
+\frac{\alpha_k^2\sigma^2}{2\mu^2}E\left[||\boldsymbol{x}_{k}||^2\right]
\ena
\eeq
Next, by dividing both sides of the inequality by $\alpha_k^2$ and using $(1-\alpha_k)/\alpha_k^2\leq1/a_{k-1}^2$ for $k\geq 1$ (see (\ref{eq:pr_le_step_size1}) in the proof of Lemma \ref{le_step_size}) we obtain:
\beq
\bea{l}
\label{eq:upper_bound_smoothing3}
\displaystyle\frac{1}{2\alpha_k^2}E\left[||\boldsymbol{\omega}_{k+1}-\boldsymbol{\omega}^*||^2\right]+\frac{1}{2\alpha_k\mu}
\left(E\left[f\left(\boldsymbol{\omega}_k\right)\right]-f\left(\boldsymbol{\omega}^*\right)\right)
\vspace{0.3cm}\\\hspace{0.0cm}\displaystyle
\leq\frac{1}{2\alpha_{k-1}^2}E\left[||\boldsymbol{\omega}_k-\boldsymbol{\omega}^*||^2\right]+
\frac{C^2}{2\mu^3(k+1)}E\left[||\boldsymbol{x}_{k}\boldsymbol{x}_{k}^T||^2\right] \vspace{0.3cm}\\\hspace{3cm}\displaystyle
+\frac{\sigma^2}{2\mu^2}E\left[||\boldsymbol{x}_{k}||^2\right]
\ena
\eeq
Next, summing (\ref{eq:upper_bound_smoothing3}) over $1, 2, ..., k$ yields:
\beq
\bea{l}
\label{eq:upper_bound_summing}
\displaystyle\frac{1}{2\alpha_k^2}E\left[||\boldsymbol{\omega}_{k+1}-\boldsymbol{\omega}^*||^2\right]
+\frac{1}{2\mu}\sum_{r=1}^{k}\frac{1}{\alpha_r}\left(E\left[f\left(\boldsymbol{\omega}_k\right)\right]-f\left(\boldsymbol{\omega}^*\right)\right)
\vspace{0.3cm}\\\hspace{0.0cm}\displaystyle
\leq\frac{1}{2}E\left[||\boldsymbol{\omega}_1-\boldsymbol{\omega}^*||^2\right]+
\frac{(H_{k+1}-1)C^2}{2\mu^3}E\left[||\boldsymbol{x}_{k}\boldsymbol{x}_{k}^T||^2\right] \vspace{0.3cm}\\\hspace{3cm}\displaystyle
+k\frac{\sigma^2}{2\mu^2}E\left[||\boldsymbol{x}_{k}||^2\right]
\vspace{0.3cm}\\\hspace{0.0cm}\displaystyle
\leq\frac{1}{2}E\left[||\boldsymbol{\omega}_1-\boldsymbol{\omega}^*||^2\right]+
\frac{\log(k+1)C^2}{2\mu^3}E\left[||\boldsymbol{x}_{k}\boldsymbol{x}_{k}^T||^2\right] \vspace{0.3cm}\\\hspace{3cm}\displaystyle
+\frac{k\sigma^2}{2\mu^2}E\left[||\boldsymbol{x}_{k}||^2\right],
\ena
\eeq
where $H_k$ is the $k^{th}$ harmonic number.

Computing the term for $k=0$ is obtained by substituting $k=0$ in (\ref{eq:upper_bound_smoothing2}):
\beq
\bea{l}
\label{eq:upper_bound_summing_k_0}
\displaystyle\frac{1}{2}E\left[||\boldsymbol{\omega}_{1}-\boldsymbol{\omega}^*||^2\right]
+\frac{1}{2\mu}\left(E\left[f\left(\boldsymbol{\omega}_0\right)\right]-f\left(\boldsymbol{\omega}^*\right)\right)
\vspace{0.3cm}\\\hspace{0.0cm}\displaystyle
\leq
\frac{C^2}{2\mu^3}E\left[||\boldsymbol{x}_{k}\boldsymbol{x}_{k}^T||^2\right] +\frac{\sigma^2}{2\mu^2}E\left[||\boldsymbol{x}_{k}||^2\right].
\ena
\eeq
As a result, by combining (\ref{eq:upper_bound_summing}) and (\ref{eq:upper_bound_summing_k_0}) we obtain for all $k\geq 0$:
\beq
\bea{l}
\label{eq:upper_bound1}
\displaystyle\frac{1}{2\alpha_k^2}E\left[||\boldsymbol{\omega}_{k+1}-\boldsymbol{\omega}^*||^2\right]
+\frac{1}{2\mu}\sum_{r=0}^{k}\frac{1}{\alpha_r}\left(E\left[f\left(\boldsymbol{\omega}_k\right)\right]-f\left(\boldsymbol{\omega}^*\right)\right)
\vspace{0.3cm}\\\hspace{0.0cm}\displaystyle
\leq
\frac{\left(\log(k+1)+1\right)C^2}{2\mu^3}E\left[||\boldsymbol{x}_{k}\boldsymbol{x}_{k}^T||^2\right] \vspace{0.3cm}\\\hspace{4cm}\displaystyle
+\frac{(k+1)\sigma^2}{2\mu^2}E\left[||\boldsymbol{x}_{k}||^2\right],
\ena
\eeq
Multiplying by $2\mu\alpha_k^2$ and rearranging terms yields:
\beq
\bea{l}
\label{eq:upper_bound2}
\displaystyle \alpha_k^2\sum_{r=0}^{k}\frac{1}{\alpha_r}\left(E\left[f\left(\boldsymbol{\omega}_k\right)\right]-f\left(\boldsymbol{\omega}^*\right)\right)
\vspace{0.3cm}\\\hspace{0.0cm}\displaystyle
\leq
\frac{\left(\log(k+1)+1\right)\alpha_k^2 C^2}{\mu^2}E\left[||\boldsymbol{x}_{k}\boldsymbol{x}_{k}^T||^2\right] \vspace{0.3cm}\\\hspace{0.5cm}\displaystyle
+\frac{(k+1)\alpha_k^2\sigma^2}{\mu}E\left[||\boldsymbol{x}_{k}||^2\right]
%\vspace{0.3cm}\\\hspace{4cm}\displaystyle
-2\mu E\left[||\boldsymbol{\omega}_{k+1}-\boldsymbol{\omega}^*||^2\right]
\vspace{0.3cm}\\\hspace{0.0cm}\displaystyle
\leq
\frac{\left(\log(k+1)+1\right)\alpha_k^2 C^2}{\mu^2}E\left[||\boldsymbol{x}_{k}\boldsymbol{x}_{k}^T||^2\right] \vspace{0.3cm}\\\hspace{3cm}\displaystyle
+\frac{(k+1)\alpha_k^2\sigma^2}{\mu}E\left[||\boldsymbol{x}_{k}||^2\right].
\ena
\eeq
Next, we use (\ref{eq:le_step_size}) in Lemma \ref{le_step_size} to get:
\beq
\bea{l}
\label{eq:upper_bound3}
\displaystyle
\frac{1}{\sum_{r=0}^{k}\frac{1}{\alpha_r}}\sum_{r=0}^{k}
\frac{1}{\alpha_r}\left(E\left[f\left(\boldsymbol{\omega}_k\right)\right]-f\left(\boldsymbol{\omega}^*\right)\right)
\vspace{0.3cm}\\\hspace{0.0cm}\displaystyle
\leq
\frac{\left(\log(k+1)+1\right)\alpha_k^2 C^2}{\mu^2}E\left[||\boldsymbol{x}_{k}\boldsymbol{x}_{k}^T||^2\right] \vspace{0.3cm}\\\hspace{3cm}\displaystyle
+\frac{(k+1)\alpha_k^2\sigma^2}{\mu}E\left[||\boldsymbol{x}_{k}||^2\right].
\ena
\eeq
Recall that $\alpha_k=\frac{\gamma}{\gamma+k}$, where $\gamma\geq 2$. Hence,
\beq
\bea{l}
\label{eq:upper_bound4}
\displaystyle
\frac{1}{\sum_{r=0}^{k}\frac{1}{\alpha_r}}\sum_{r=0}^{k}
\frac{1}{\alpha_r}\left(E\left[f\left(\boldsymbol{\omega}_k\right)\right]-f\left(\boldsymbol{\omega}^*\right)\right)
\vspace{0.3cm}\\\hspace{0.0cm}\displaystyle
\leq
\frac{\left(\log(k+1)+1\right)\gamma^2 C^2}{\mu^2(\gamma+k)^2}E\left[||\boldsymbol{x}_{k}\boldsymbol{x}_{k}^T||^2\right] \vspace{0.3cm}\\\hspace{3cm}\displaystyle
+\frac{(k+1)\gamma^2\sigma^2}{\mu(\gamma+k)^2}E\left[||\boldsymbol{x}_{k}||^2\right].
\ena
\eeq
Next, using the convexity of $f$ we have:
\beq
\bea{l}
\label{eq:upper_bound5}
\displaystyle
E\left[f\left(\bar{\boldsymbol{\omega}}_k\right)\right]-f\left(\boldsymbol{\omega}^*\right)
\vspace{0.3cm}\\\hspace{0.0cm}\displaystyle
\leq
\frac{\left(\log(k+1)+1\right)\gamma^2 C^2}{\mu^2(\gamma+k)^2}E\left[||\boldsymbol{x}_{k}\boldsymbol{x}_{k}^T||^2\right] \vspace{0.3cm}\\\hspace{2cm}\displaystyle
+\frac{(k+1)\gamma^2\sigma^2}{\mu(\gamma+k)^2}E\left[||\boldsymbol{x}_{k}||^2\right], \;\;\forall \gamma\geq2 \;\forall k\geq 0. \vspace{0.3cm}
\ena
\eeq
\end{proof}

%---------------------remark-----------------------
\begin{remark}
From Theorem \ref{th:main}, we obtain an explicit $O(1/k)$ upper bound on the convergence rate, depending on the noise variance (second term on the RHS of (\ref{eq:th_mean_squares1})) and the size of the constraint set (first term on the RHS of (\ref{eq:th_mean_squares1})). The variance term dominates the error and decreases with rate $1/k$, while the other term (which is related to the diameter $e_{max}$ of the constraint set) decreases faster at rate $\log(k)/k^2$. The best asymptotic (as $k$ increases) bound is obtained by setting $\gamma=2$. \vspace{0.3cm}
\end{remark}

\begin{remark}
Note that when the random components of $\boldsymbol{x}$ are identically distributed and uncorrelated (thus, the correlation matrix of $\boldsymbol{x}_{k}$ can be written as $E\left[\boldsymbol{x}\boldsymbol{x}^T\right]=\mu I_d$, where $I_d$ is the identity matrix and its minimal eigenvalue is $\mu$) we obtain: $E\left[||\boldsymbol{x}||^2\right]=d\mu$. As a result, we have $\lim_{k\rightarrow\infty}k\left(E\left[f\left(\bar{\boldsymbol{\omega}}_k\right)\right]-f\left(\boldsymbol{\omega}^*\right)\right)\leq 4d\sigma^2$, where $\lim_{k\rightarrow\infty}k\left(E\left[f\left(\boldsymbol{\omega}^{ERM}_k\right)\right]-f\left(\boldsymbol{\omega}^*\right)\right)=d\sigma^2$ under the ERM scheme. Hence, the asymptotic ratio $\rho$ between the convergence rate of our scheme and the ERM scheme is upper bounded by $\rho\leq 4$ as the number of iterations approaches infinity.
\vspace{0.3cm}
\end{remark}

\begin{remark}
The streaming SVRG algorithm proposed in \cite{frostig2014competing} for a general strongly convex regression problem achieves $\rho=1$ asymptotically with the price of geometrically increasing batch sample size occasionally and setting the constant step size close to zero, which deteriorates performance in the finite regime. The SGD with averaging and constant step size scheme proposed in \cite{defossez2015averaged} for a linear least squares regression problem requires a fixed batch sample size as required by PSGD-WA. However, obtaining $\rho=1$ asymptotically requires to set the constant step size close to zero, which deteriorates performance in the finite regime (due to a term that depends on $1/\zeta^2$ and blows up as the constant step size $\zeta$ approaches zero \cite{defossez2015averaged}). Controlling the decay step sizes, however, as suggested by PSGD-WA avoids that blowing up term. Theorem \ref{th:main} shows that PSGD-WA achieves $\rho\leq 4$, where the step sizes can be large in the beginning of the algorithm and approach zero only asymptotically. This insight is demonstrated by numerical experiments in Section \ref{sec:simulation}, where significant performance gain is demonstrated by PSGD-WA in the finite regime, while unweighted averaging is expected to perform well as the number of iterations becomes very large.
\end{remark}

\subsection{A case of $d=1$}
\label{sec:unconstrained}

For purposes of analysis whether further improvement in the resulting error can be expected, we provide a better bound for the error when $d=1$, and $\Omega$ is unbounded.

Let
\beq
\label{eq:step_size_unconstrained}
\displaystyle\lambda_k=\frac{1}{2x_k^2}\alpha_k=\frac{1}{2x_k^2}\frac{\gamma}{\gamma+k}.
\eeq
Note that when $x_k=x$ for all $k$, then $\mu=x^2$, and $\lambda_k=\frac{1}{2x^2}\alpha_k=\frac{1}{2x^2}\frac{\gamma}{\gamma+k}$, which is a special case of the step size in (\ref{eq:th_lambda_beta}) when $d=1$.

Since $\Omega$ is unbounded, the proposed PSGD-WA algorithm updates the estimate $\omega_{k+1}$ using a SGD update and compute a weighted average over iterates as byproduct of the algorithm. Specifically, at iteration $k+1$ we compute the estimate of $\omega^*$ as follows:
\beq
\label{eq:SGD_WA1_unconstrained}
\bea{l}
\displaystyle\omega_{k+1}=\omega_{k}-\alpha_k\left(\omega_{k}-y_{k}/x_k\right),
\ena
\eeq
where $\alpha_k=\frac{\gamma}{\gamma+k}$.
In addition to the estimate $\omega_{k+1}$, we compute the weighted average estimate as in (\ref{eq:PSGD_WA2}).

Let
\beq
\bea{l}
\displaystyle\eta_k\triangleq\omega_k-\omega^*, \vspace{0.3cm}\\
\displaystyle\bar{\eta}_k\triangleq\bar{\omega}_k-\omega^*=\sum_{i=0}^{k}\beta_{k,i}\eta_i.
\ena
\eeq
where the last equality holds since $\sum_{i=0}^{k}\beta_{k,i}=1$ for all $k\geq 0$.

We define:
\beq
\bea{l}
\displaystyle\widetilde{M}_{i,j}\triangleq\prod_{r=i+1}^{j}\left(1-\alpha_{r-1}\right),
\ena
\eeq
and
\beq
\bea{l}
\displaystyle M_{i,j}\triangleq\beta_j\widetilde{M}_{i,j}.
\ena
\eeq

For the ease of presentation we also set\footnote{It should be noted that a similar asymptotic result in this section is obtained by setting $\beta_{k,i}$ as in (\ref{eq:th_lambda_beta})}:
\beq
\bea{l}
\label{eq:conditions}
%\displaystyle\lambda_k=\frac{\gamma}{\gamma+k} \; \forall k\geq 0,\vspace{0.3cm}\\
\displaystyle\beta_0=1\;,\;\beta_k=\frac{1}{\alpha_{k-1}}=\frac{\gamma+k-1}{\gamma}\;\forall k\geq 1\vspace{0.3cm}\\
\displaystyle\beta_{k,i}=\frac{\beta_i}{\sum_{r=0}^k \beta_r} \; \forall k\geq 0.
\ena
\eeq
where $\gamma\geq 1$.

%------------lemma------------
\begin{lemma}
\label{lemma:sum_M}
Assume that (\ref{eq:conditions}) holds. Then,
\beq
\label{eq:sum_M}
\displaystyle\sum_{j=i+1}^{k}M_{i+1,j}\leq\frac{(i+\gamma)(k-i)}{\gamma}.
\eeq
\end{lemma}
%-----------proof--------------
\begin{proof}
Since $1-\alpha_i=1-\frac{\gamma}{\gamma+i}=\frac{i}{\gamma+i}$, we can rewrite $M_{i+1,j}$ as:
\beq
\bea{l}
\displaystyle M_{i+1,j}=\frac{\gamma+j-1}{\gamma}\prod_{r=i+2}^{j}\frac{r-1}{\gamma+r-1}\;\forall j\geq i+1.
\ena
\eeq
As a result, we obtain:
\beq
\bea{l}
\displaystyle M_{i+1,j}=\frac{\gamma+j-1}{\gamma}\times\vspace{0.3cm}\\\hspace{0.3cm}
\displaystyle\left[\frac{i+1}{\gamma+i+1}\cdot\frac{i+2}{\gamma+i+2}
                   \cdots\frac{j-2}{\gamma+j-2}\cdot\frac{j-1}{\gamma+j-1}\right]\vspace{0.3cm}\\\hspace{0.3cm}
\displaystyle\leq\frac{\gamma+j-1}{\gamma}\times\vspace{0.3cm}\\\hspace{0.3cm}
\displaystyle\left[(i+1)\cdots(i+\lfloor\gamma\rfloor)\frac{i+\lfloor\gamma\rfloor+1}{\gamma+i+1}\cdots\frac{j-1}{j-1+\gamma-\lfloor\gamma\rfloor} \times\right.\vspace{0.3cm}\\\hspace{0.3cm}
\displaystyle\left.
\frac{1}{j+\gamma-\lfloor\gamma\rfloor}\cdots\frac{1}{j+\gamma-1}\right]
\vspace{0.3cm}\\\hspace{0.3cm}
\displaystyle\leq\frac{\gamma+j-1}{\gamma}\left[\frac{i+1}{j+\gamma-\lfloor\gamma\rfloor}\cdots\frac{i+\lfloor\gamma\rfloor}{j+\gamma-1}\right]
\vspace{0.3cm}\\\hspace{0.3cm}
\displaystyle\leq\frac{\gamma+j-1}{\gamma}
                    \left[\frac{i+\lfloor\gamma\rfloor}{j+\gamma-1}\right]\leq\frac{i+\gamma}{\gamma}.
\ena
\eeq
Thus, summing over $j$ yields (\ref{eq:sum_M}).
\end{proof}

%---------------------theorem-----------------------
\begin{theorem}
\label{th:main_unconstrained}
Assume that PSGD-WA is implemented, where the parameters satisfies (\ref{eq:conditions}). Then, \vspace{0.3cm}\\
a) for all $k\geq 0$ we have:
\beq
\bea{l}
\label{eq:th_mean_squares2}
\displaystyle
E\left[f\left(\bar{\omega}_k\right)\right]-f\left(\omega^*\right)
\leq
\frac{4\gamma^2\sigma^2E[x^2]E[1/x^2]}{3k}+O\left(k^{-2}\right).
\ena
\eeq
b) In addition, if $x_k=x$ for all $k$ and $\gamma=1$, we have:
\beq
\bea{l}
\label{eq:th_mean_squares2}
\displaystyle
\lim_{k\rightarrow\infty}k\left(E\left[f\left(\bar{\omega}_k\right)\right]-f\left(\omega^*\right)\right)\leq \frac{4}{3}\sigma^2. \vspace{0.3cm}
\ena
\eeq
\end{theorem}
%--------------proof-------------------
\begin{proof}
Since we consider a least squares loss, we have:
\beq
\bea{l}
\displaystyle
E\left[f\left(\bar{\omega}_k\right)\right]-f\left(\omega^*\right)
=E[x^2]E\left[\bar{\eta}_k^2\right],
\ena
\eeq

Next, we compte $\bar{\eta}_k$. Note that $\eta_i$ can be written recursively as follows:
\beq
\bea{l}
\displaystyle\eta_i=\omega_i-\omega^*\vspace{0.3cm}\\
\displaystyle=\omega_{i-1}-\alpha_{i-1}(\omega_{i-1}-y_{i-1}/x_{i-1})-\omega^*\vspace{0.3cm}\\
\displaystyle=(1-\alpha_{i-1})\eta_{i-1}+\alpha_{i-1}v_{i-1}/x_{i-1},
\ena
\eeq
and by iterating over $\eta_i$ we obtain:
\beq
\bea{l}
\displaystyle\eta_k=\sum_{i=0}^{k-1}\widetilde{M}_{i+1,k}\alpha_iv_i/x_i.
\ena
\eeq
Hence,
\beq
\bea{l}
\displaystyle\bar{\eta}_k
=\frac{1}{\sum_{r=0}^{k}\beta_r}\sum_{j=0}^{k}\sum_{i=0}^{j-1}\beta_j\widetilde{M}_{i+1,j}\alpha_iv_i/x_i \vspace{0.3cm}\\
\displaystyle=\frac{1}{\sum_{r=0}^{k}\beta_r}\sum_{i=0}^{k-1}\left(\sum_{j=i+1}^{k}M_{i+1,j}\right)\alpha_iv_i/x_i.
\ena
\eeq

Next, we compute $E\left[\bar{\eta}_k^2\right]$. Note that:
\beq
\bea{l}
\displaystyle E\left[\bar{\eta}_k^2\right]
=\frac{1}{\left(\sum_{r=0}^{k}\beta_r\right)^2}\times\vspace{0.3cm}\\
\displaystyle E\left[\sum_{i=0}^{k-1}\left(\sum_{j=i+1}^{k}M_{i+1,j}\right)\alpha_i\frac{v_i}{x_i}
       \sum_{\ell=0}^{k-1}\left(\sum_{p=\ell+1}^{k}M_{\ell+1,p}\right)\alpha_{\ell}\frac{v_{\ell}}{x_{\ell}}\right].
\ena
\eeq
Since cross terms are canceled (due to independence across time), we obtain:
\beq
\bea{l}
\displaystyle E\left[\bar{\eta}_k^2\right]%\vspace{0.3cm}\\\hspace{0.5cm}
\displaystyle=\frac{\sigma^2E[1/x^2]}{\left(\sum_{r=0}^{k}\beta_r\right)^2}%\times\vspace{0.3cm}\\
\displaystyle \sum_{i=0}^{k-1}\left(\sum_{j=i+1}^{k}M_{i+1,j}\right)^2\alpha_i^2.
\ena
\eeq

Setting $\beta_r$ according to (\ref{eq:conditions}) yields:
\beq
\label{eq:sum_beta}
\left(\sum_{r=0}^{k}\beta_r\right)^2=\frac{k^4}{4}+O(k^3).
\eeq

Next, applying Lemma \ref{lemma:sum_M} and setting $\alpha_i=\frac{\gamma}{\gamma+i}$ yields:
\beq
\bea{l}
\displaystyle E\left[\bar{\eta}_k^2\right]%\vspace{0.3cm}\\\hspace{0.5cm}
\displaystyle\leq\frac{\gamma^2\sigma^2E[1/x^2]}{k^4/4+O(k^3)}%\times\vspace{0.3cm}\\
\displaystyle \sum_{i=0}^{k-1}\left(k-i\right)^2.
\ena
\eeq
Finally, since $\sum_{i=0}^{k-1}\left(k-i\right)^2=k^3/3+O(k^2)$, (\ref{eq:th_mean_squares2}) follows. Setting $\gamma=1$, $x_k=x$ for all $k$, and letting $k\rightarrow\infty$ yields (\ref{eq:th_mean_squares2}). \vspace{0.3cm}
\end{proof}

\begin{remark}
Note that when the conditions in Theorem \ref{th:main_unconstrained}.b hold, then the asymptotic ratio $\rho$ between the convergence rate of PSGD-WA and the ERM scheme is upper bounded by $\rho\leq 4/3$ as the number of iterations approaches infinity. Thus, the upper bound on the error is better then $\rho\leq 4$ obtained in Theorem \ref{th:main}. Simulation results demonstrate $\rho\leq 4/3$ even for large $d$ in practice.
\end{remark}

\section{Numerical Examples}
\label{sec:simulation}
In this section, we provide numerical examples to illustrate the performance of the algorithms. We have performed experiments on synthetic as well as real date set.

\subsection{Experiments Over Synthetic Data}
\label{ssec:synthetic}

In this section we examined the performance of the algorithms over synthetic data. We set the following parameters (very similar to the experiment setup in \cite{defossez2015averaged}): $d=25$, the streaming data $\boldsymbol{x}_k\in\mathbb{R}^{25}$ are i.i.d r.v. drawn from a normal distribution with covariance matrix $I_d$, and $y_k=\boldsymbol{x}_k^T\omega^*+v_k$, where $v_k\sim N(0,\sigma^2)$ is an additive Gaussian noise. $\boldsymbol{\omega}^*=[1 \; 2 \;...\; 25]^T$ is the unknown parameter. The constraint set for the projected SGD iterates was set to $\boldsymbol{\omega}^*\pm 100$.

We compared three streaming algorithms that require a very similar computational complexity and tuned their parameters: i) a standard Projected SGD with decreasing step size $10/(10+k)$, referred to as PSGD; ii) a Projected SGD using a constant step size $0.002$ with Averaging, referred to as PSGD-A (i.e., a projected version of the algorithm proposed in \cite{defossez2015averaged}); iii) the proposed Projected SGD algorithm with decreasing step size $10/(10+k)$ and Weighted Averaging (PSGD-WA). We performed $1000$ Monte-Carlo experiments to compute the average performance. As a benchmark, we computed the empirical risk minimizer (ERM), which solves (\ref{eq:opt_LS_intro}) directly by using the \emph{entire data} at each iteration.

First, we set $\sigma^2=0.1$. The performance of the algorithms are presented in Fig. \ref{fig:fig1}. It can be seen that the proposed PSGD-WA algorithm performs the best among the streaming algorithms and obtains performance close to the ERM algorithm for a large range of tested $k$. The ratio between the errors under PSGD-WA algorithm and the ERM schemes was less than $1.335$ for all $k> 2\cdot 10^4$ and equals $1.31$ for $k=10^5$. These results coincide with the upper bound $\rho\leq 4/3$ obtained in Theorem \ref{th:main_unconstrained} under $d=1$. However, showing $\rho\leq 4/3$ theoretically for $d>1$ remains open. It can also be seen that PSGD-A has the largest decreasing rate, thus, expected to perform well for very large $k$. These results confirm the advantages of the proposed PSGD-WA algorithm in the finite sample regime, as well as demonstrating its nice asymptotic property (up to a constant ratio between the asymptotic errors under PSGD-WA algorithm and the ERM schemes).

Next, we set $\sigma^2=1$. The performance of the algorithms are presented in Fig. \ref{fig:fig2}. It can be seen that the proposed PSGD-WA algorithm performs the best among the streaming algorithms and obtains performance close to the ERM algorithm for all tested $k$. The ratio between the errors under PSGD-WA and the ERM schemes was less than $1.332$ for all $k> 2\cdot 10^4$ and equals $1.29$ for $k=10^5$. Again, the results coincide with the upper bound $\rho\leq 4/3$ obtained in Theorem \ref{th:main_unconstrained} under $d=1$. It can also be seen that PSGD-A outperforms the standard PSGD for $k> 6\cdot 10^4$, and has the largest decreasing rate, thus, expected to perform well for very large $k$. Again, the results confirm the advantages of the proposed PSGD-WA algorithm in the finite sample regime, as well as demonstrating its nice asymptotic property. It should be noted that similar results have been observed under many different scenarios on the synthetic data.

\begin{figure}[htbp]
\centering \epsfig{file=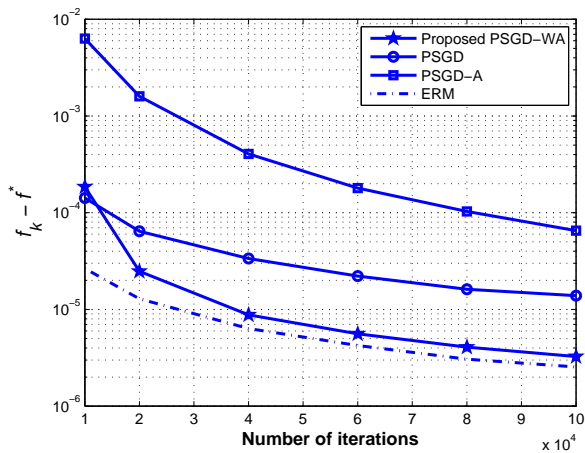,
width=0.51\textwidth}
\caption{The error as a function of the number of iterations under various PSGD algorithms as described in Sec. \ref{ssec:synthetic}.}
\label{fig:fig1}
\end{figure}

\begin{figure}[htbp]
\centering \epsfig{file=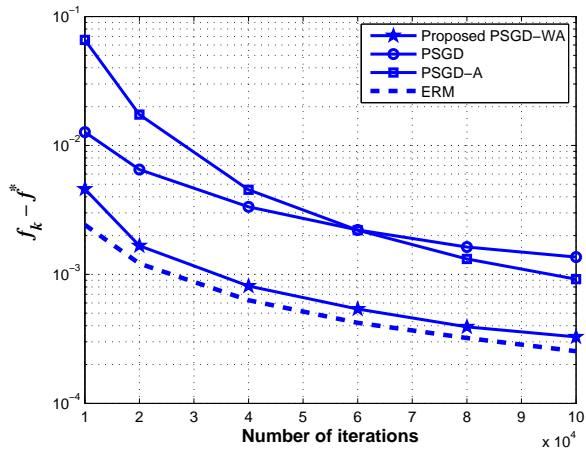,
width=0.51\textwidth}
\caption{The error as a function of the number of iterations under various PSGD algorithms as described in Sec. \ref{ssec:synthetic}.}
\label{fig:fig2}
\end{figure}

\subsection{Experiments Over the Million Song Dataset}
\label{ssec:real}

In this section we examined the performance of the algorithms for prediction of a release year of a song from audio features. We used the dataset available by UCI Machine Learning Repository \cite{Lichman:2013}, extracted from the Miliion Song Dataset collaborative project between The Echo Nest and LabROSA \cite{Bertin-Mahieux2011}. The Million Song Dataset contains songs which are mostly western, commercial tracks ranging from 1922 to 2011. Each song is associated with a released year (i.e., $y$ in our model that we aim to estimate), and $90$ audio attributes (i.e., $\boldsymbol{x}$ in our model).
We compared three streaming algorithms as described in Sec. \ref{ssec:synthetic}. Here, we did not assume prior knowledge on the constraint set of the parameters. Thus, the projected update degenerates to an unconstrained update:
\begin{center}
$
\displaystyle\boldsymbol{\omega}_{k+1}=\boldsymbol{\omega}_{k+1}-\lambda_k\boldsymbol{g}_{k}.
$
\end{center}

In Fig. \ref{fig:fig3}, we present the average prediction error of the released year of a song $|\hat{y}_k-y_k|$ as a function of the number of iterations. It can be seen that the proposed PSGD-WA algorithm performs the best among the streaming algorithms for all tested $k$. It can also be seen that the standard PSGD performs the worst for all tested $k$. It should be noted that the simulation results demonstrate that PSGD-A has high decreasing rate, thus, expected to perform well as $k$ becomes large. In Fig. \ref{fig:fig4}, we present the average normalized (i.e., the range [1922, 2011] was mapped to [0,1]) prediction square error of the released year of a song as a function of the number of iterations. It can be seen that the proposed PSGD-WA algorithm performs the best among the streaming algorithms for $k<4\cdot 10^5$, where PSGD-A algorithm performs the best for $k\geq 4\cdot 10^5$, thus, expected to perform well as $k$ becomes very large. These results confirm the advantages of the proposed PSGD-WA algorithm in the finite sample regime and provide important design principles when implementing PSGD algorithms for regression tasks.

\begin{figure}[htbp]
\centering \epsfig{file=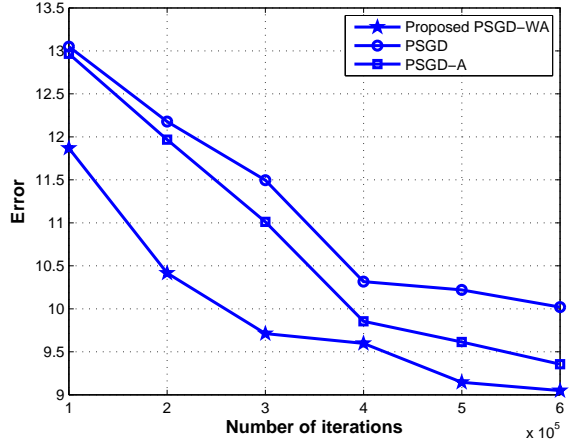,
width=0.51\textwidth}
\caption{average prediction error of the released year of a song $|\hat{y}_k-y_k|$ as a function of the number of iterations under various PSGD algorithms as described in Sec. \ref{ssec:real}.}
\label{fig:fig3}
\end{figure}

\begin{figure}[htbp]
\centering \epsfig{file=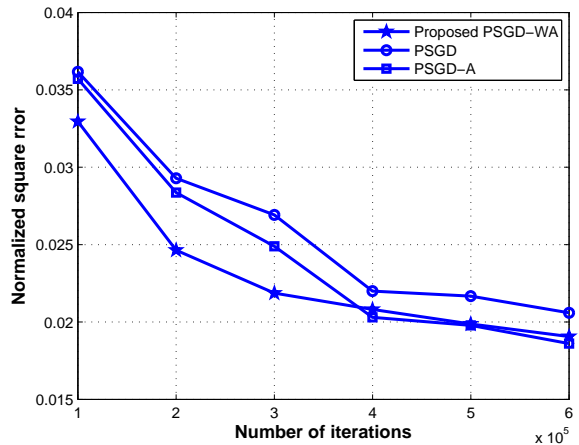,
width=0.51\textwidth}
\caption{average normalized (i.e., the range [1922, 2011] was mapped to [0,1]) prediction square error of the released year of a song as a function of the number of iterations under various PSGD algorithms as described in Sec. \ref{ssec:real}.}
\label{fig:fig4}
\end{figure}

\section{Conclusion}
\label{sec:conclusion}
We considered a least squares regression of a $d$-dimensional unknown parameter. We proposed and analyzed a stochastic gradient descent algorithms with weighted iterate-averaging that uses a single pass over the data. When the constraint set of the unknown parameter is bounded, we provided an explicit $O(1/k)$ upper bound on the convergence rate, showing that the variance term dominates the error and decreases with rate $1/k$, while the term which is related to the size of the constraint set decreases with rate $\log k/k^2$. We then compared the asymptotic ratio $\rho$ between the convergence rate of the proposed scheme and the empirical risk minimizer (ERM) as the number of iterations approaches infinity. Under some mild conditions, we showed that $\rho\leq 4$ for all $d\geq 1$. We further improved the upper bound by showing that $\rho\leq 4/3$ for the case of $d=1$ and when the parameter set is unbounded.

Simulation results over synthetic data demonstrate strong performance of the algorithm as compared to existing methods, and coincide with $\rho\leq 4/3$ even for large $d$ in practice. We also tested the algorithm over the Million Song Dataset and strong performance has been obtained as compared to existing methods under the finite sample regime

It should be noted that SGD with a constant step size does not converge to the global optimum in general \cite{nedic2001convergence, bach2013non}. Thus, it is desirable to analyze the proposed PSGD-WA algorithm with decreasing step size under other loss functions (e.g., logistic regression) as a future research direction.

\newcommand*{\QEDA}{\hfill\ensuremath{\blacksquare}}%
\QEDA

\bibliographystyle{ieeetr}
%\bibliography{Co_Ne_Sr_PSGD_TAC_bib}

%
\end{document}

%% file: macros.tex
\newtheorem{theorem}{Theorem}
\newtheorem{lemma}{Lemma}
\newtheorem{remark}{Remark}
\newtheorem{definition}{Definition}

\newcommand{\beq}{\begin{equation}}
\newcommand{\eeq}{\end{equation}}
\newcommand{\bea}{\begin{array}}
\newcommand{\ena}{\end{array}}
\newcommand{\bds}{\begin {itemize}}
\newcommand{\eds}{\end {itemize}}
\newcommand{\bdf}{\begin{definition}}
\newcommand{\blm}{\begin{lemma}}
\newcommand{\edf}{\end{definition}}
\newcommand{\elm}{\end{lemma}}
\newcommand{\bthm}{\begin{theorem}}
\newcommand{\ethm}{\end{theorem}}
\newcommand{\bprp}{\begin{prop}}
\newcommand{\eprp}{\end{prop}}
\newcommand{\bcl}{\begin{claim}}
\newcommand{\ecl}{\end{claim}}
\newcommand{\bcr}{\begin{coro}}
\newcommand{\ecr}{\end{coro}}
\newcommand{\bquest}{\begin{question}}
\newcommand{\equest}{\end{question}}

\newcommand{\larrow}{{\larrow}}

\def\urltilda{\kern -.15em\lower .7ex\hbox{\~{}}\kern .04em}